\documentclass[12pt]{amsart}

\usepackage{verbatim,amsmath,latexsym,amssymb,amsbsy}
\usepackage[all]{xy}

\theoremstyle{plain}
\newtheorem{thm}{Theorem}[section]

\newtheorem{lem}[thm]{Lemma}
\newtheorem{prop}[thm]{Proposition}

\theoremstyle{definition}

\newtheorem{df}[thm]{Definition}

\newtheorem{ex}[thm]{Example}
\newtheorem{ex-notn}[thm]{Example/Notation}

\newtheorem{rem}[thm]{Remark}

\def\ZZ{{\mathbb Z}}

\def\ra{\rightarrow}

\def\Mtwo{{\em Macaulay} 2\expandafter}

\numberwithin{equation}{section}
\begin{document}

\title{On the minimality of Hamming compatible metrics}
\author{Parsa Bakhtary, Othman Echi}

\address{Department of Mathematics and Statistics\\
King Fahd University of Petroleum and Minerals\\
Dhahran, Saudi Arabia 31260}
\begin{abstract}A Hamming compatible metric is an integer-valued metric on the words of a finite alphabet which agrees with
the usual Hamming distance for words of equal length.  We define a new Hamming compatible metric, compute the
cardinality of a sphere with respect to this metric, and show this metric is minimal in the class of all ``well-behaved"
Hamming compatible metrics.
\end{abstract}
\thanks{}
\email{pbakhtary@kfupm.edu.sa, echi@kfupm.edu.sa}

\maketitle
\section{Introduction}
Ever since Richard Hamming's seminal 1950 paper \cite{H}, the notion of the Hamming distance has played a fundamental role in the development of coding theory, error-correcting codes, cryptography, telecommunication, and information theory.  Simply put, the Hamming distance between two words (or strings) of equal length counts the number of places where the corresponding letters differ.  For instance `coma' and `comb' have Hamming distance $1$, while `sunny' and `burnt' have Hamming distance $3$, and it is easy to check that this is in fact a metric on the set of words of a given length.

However, the classical Hamming distance is restrictive in that it only measures the distance between two words of equal length.  One would also like to have a similar metric for words of different lengths which agrees with the Hamming distance when those words are of the same length.  To this end, the second author introduced \cite{E} the notion of an integer-valued  {\emph{Hamming compatible}} metric and gave a natural example.

A natural question that arises is ``How small can such a metric be?"  Using only the axioms of a metric and the fact that it is integer-valued and Hamming compatible, one may try to find the smallest such metric on a language.  Unfortunately, these properties alone are not enough to say something substantive about the minimality of such metrics, and we will see there are some examples which do not have the desired behavior.  On the other hand, if we make some natural uniformity assumptions on the metric, we can show that there is indeed a smallest such metric.

In what follows, we fix a finite set $\Sigma$, called an alphabet,
and denote by $\Sigma_n$ the set of $n$-letter words. The collection
of all words of finite length is denoted $\Sigma^*$, called the
improper language.  All metrics are assumed to be integer-valued.

\section{The $d_2$ metric}
We begin with the following characterization of Hamming distance, $H$.

\begin{prop}Let $d: \Sigma^* \times \Sigma^* \ra \ZZ_{\geq 0}$ be a
mapping such that the induced map Suppose $d_n$ on $\Sigma_n \times
\Sigma_n $ is a metric for every $n \in \ZZ_{\geq 0}$.  Then $d=H$
is the Hamming distance on $\Sigma_n$ for every $n \in \ZZ_{\geq 0}$
if and only if
\begin{enumerate}
\item for all words $u,v$ with $n=l(u)=l(v)$ we have $d(u,v) \leq n$, and
\item $d(u_1 v_1,u_2 v_2) = d(u_1,u_2) + d(v_1,v_2)$ whenever $l(u_1)=l(u_2)$ and $l(v_1)=l(v_2)$.
\end{enumerate}
\end{prop}

\begin{proof}
It is obvious by the definition that the Hamming distance satisfies properties (1) and (2).  To show the converse, we induct
on $n=l(u)=l(v)$.  For $n=1$, $u=a$ and $v=b$ are phonemes.  Property (1) implies that $d(a,b) \leq 1$ and can only be $0$ if $a=b$ since $d$ is a metric.  Hence, if $a \neq b$, we must have $d(a,b)=1$.  So $d$ is the Hamming distance for phonemes, i.e. for words of length $1$.

For the inductive step let $u', v'$ be words of length $n+1$ and write $u'=au, v'=bv$ where $a,b$ are phonemes and $u,v$ are words of length $n$.  Using property (2), the $n=1$ case, and the inductive hypothesis we have
\[ d(u',v') = d(a,b) + d(u,v) = H(a,b) + H(u,v) = H(u',v'). \]
Hence, $d$ is the Hamming distance.
\end{proof}

\begin{df}
A metric $d$ on the language $\Sigma^*$ is called \emph{Hamming
compatible} if for any $n \in \ZZ_{\geq 0}$, $d(u,v)=H(u,v)$ for all
$u,v \in \Sigma_n$.
\end{df}

Let $H(\underline{u},\underline{v})$ be the truncated Hamming
function, defined as follows.  If $l(u) \geq l(v)$, then
$\underline{u}$ drops the last $l(u)-l(v)$ letters of $u$ and
$\underline{v}=v$, so that
$H(\underline{u},\underline{v})=H(\underline{u},v)$ is the usual
Hamming distance between two words of length $l(v)$.  Observe that
the truncated Hamming function is not a metric!

\begin{ex}
In \cite{E}, the second author defined $T(u,v):=H(\underline{u},\underline{v}) + |l(u)-l(v)|$, and showed that it is a
Hamming compatible metric.  It is easily seen to be Hamming compatible, and one checks the triangle inequality by exhausting the cases.
\end{ex}

Define the following:  \[ d_n(u,v)= H(\underline{u},\underline{v})+ \gamma_n(u,v) \]

where \[ \gamma_n(u,v) = \left \{
\begin{array} {c l}
\frac{|l(u)-l(v)|}{n} & \text{if $l(u)-l(v)\equiv 0\pmod n$} \\
\\
\frac{|l(u)-l(v)|+n-i}{n} & \text{if $l(u)-l(v)\equiv i\pmod n$, for
some $i\in\{0,1,\ldots, n-1\}$}
\end{array}  \right. \]

for words $u,v$ in the language $\Sigma^*$ over an alphabet
$\Sigma$. Let $N=|\Sigma|$.

If $n \geq 3$, then $d_n$ is not a metric, as the following example shows.  Let $\Sigma=\{0,1\}$, $w=0^n$, $u=0^{2n}$,
and $v=0^n 1^n 0^n$.  Here we mean that $w$ is the $n$-letter word consisting solely of zeros, and so on.

Then, \\ $d_n(u,v) = n + \frac{n}{n} = n+1$ \\
$d_n(u,w) = 0 + \frac{n}{n} = 1$ \\
$d_n(v,w) = 0 + \frac{n+n}{n} = 2$, so we have that
$d_n(u,v) > d_n(u,w) + d_n(v,w)$.

However, for $n=2$, this is an integer-valued, Hamming compatible {\emph{metric}}.  In this case,

 \[ d_2(u,v)= H(\underline{u},\underline{v})+ \gamma_2(u,v) \]

where \[ \gamma_2(u,v) = \left \{
\begin{array} {c l}
\frac{|l(u)-l(v)|}{2} & \text{if $l(u)-l(v)$ is even} \\
\\
\frac{|l(u)-l(v)|+1}{2} & \text{if $l(u)-l(v)$ is odd}
\end{array}  \right. \]

Another way to write this is to say $\gamma_2(u,v) = \lceil \frac{|l(u)-l(v)|}{2} \rceil$, i.e. the least integer that is greater than or equal to $\frac{|l(u)-l(v)|}{2}$, also called the round up, or ceiling function.

\begin{prop}
The function $d_2=H(\underline{u},\underline{v})+\lceil \frac{|l(u)-l(v)|}{2} \rceil $ is a metric on $\Sigma^*$.
\end{prop}

\begin{proof}
We check the triangle inequality for $\delta = d_2$.  Let $u,v,w \in \Sigma^*$, and set $D = \delta(u,w) + \delta(w,v) -\delta(u,v)$.  We must show $D \geq 0$ in all cases.

\textbf{Case 1:} $l(w) \leq l(u) \leq l(v)$

Then we may write $u=u_1 u_2, v=v_1 v_2 v_3$ where $l(w) = l(u_1) = l(v_1)$ and $l(u_2) = l(v_2)$.  We also have \\
$\delta(u,v) = H(u_1,v_1) + H(u_2,v_2) + \frac{l(v_3) + \alpha_1}{2}$ \\
$\delta(u,w) = H(u_1,w) + \frac{l(u_2) + \alpha_2}{2}$ \\
$\delta(w,v) = H(v_1,w) + \frac{l(v_2) + l(v_3) + \alpha_3}{2}$ , where each $\alpha_i \in \{0,1\}$.  Then we have
\[ D=[H(u_1,w) + H(v_1,w) - H(u_1,v_1)] + [l(u_2) - H(u_2,v_2)] + \frac{\alpha_2 + \alpha_3 - \alpha_1}{2}. \]
Since $H$ is a metric on $\Sigma_n$ for $n=l(w)$, the first term is nonnegative.  The second term is nonnegative, as it is
just property (1) from the previous proposition.  It suffices to show the third term $L= \frac{\alpha_2 + \alpha_3 - \alpha_1}{2}$ is nonnegative.

\emph{Subcase 1A:} $l(u) - l(v)$ is even.  Then $l(u)-l(w)$ and $l(v)-l(w)$ are either both even or both odd.  If both are even, then
$\alpha_1=\alpha_2=\alpha_3=0$, so $L=0 \geq 0$.  If both are odd, then $\alpha_1=0, \alpha_2=\alpha_3=1$ so $L=1 \geq 0$.

\emph{Subcase 1B:} $l(u) - l(v)$ is odd.  Then exactly one of $l(u)-l(w)$ and $l(v)-l(w)$ is odd, and the other is even.  Hence $\alpha_1=1$ and $\{\alpha_2,\alpha_3\} = \{0,1\}$, so $L=0 \geq 0$.

\textbf{Case 2:} $l(u) \leq l(w) \leq l(v)$

Then we may write $w=w_1 w_2, v=v_1 v_2 v_3$ where $l(u) = l(w_1) = l(v_1)$ and $l(w_2) = l(v_2)$.  We also have \\
$\delta(u,v) = H(u,v_1) + \frac{l(v_2) + l(v_3) + \alpha_1}{2}$ \\
$\delta(u,w) = H(u,w_1) + \frac{l(w_2) + \alpha_2}{2}$ \\
$\delta(w,v) = H(w_1,v_1) + H(w_2,v_2) + \frac{l(v_3) + \alpha_3}{2}$ , where each $\alpha_i \in \{0,1\}$.  Then we have
\[ D=[H(u,w_1) + H(w_1,v_1) - H(u,v_1)] + H(w_2,v_2) + \frac{\alpha_2 + \alpha_3 - \alpha_1}{2}. \]
Again, it suffices to show the third term $L= \frac{\alpha_2 + \alpha_3 - \alpha_1}{2}$ is nonnegative.

\emph{Subcase 2A:} $l(u) - l(v)$ is even.  Then $l(u)-l(w)$ and $l(v)-l(w)$ are either both even or both odd.  If both are even, then
$\alpha_1=\alpha_2=\alpha_3=0$, so $L=0 \geq 0$.  If both are odd, then $\alpha_1=0, \alpha_2=\alpha_3=1$ so $L=1 \geq 0$.

\emph{Subcase 2B:} $l(u) - l(v)$ is odd.  Then exactly one of $l(u)-l(w)$ and $l(v)-l(w)$ is odd, and the other is even.  Hence $\alpha_1=1$ and $\{\alpha_2,\alpha_3\} = \{0,1\}$, so $L=0 \geq 0$.

\textbf{Case 3:} $l(u) \leq l(v) \leq l(w)$

Then we may write $v=v_1 v_2, w=w_1 w_2 w_3$ where $l(u) = l(v_1) = l(w_1)$ and $l(v_2) = l(w_2)$.  We also have \\
$\delta(u,v) = H(u,v_1) + \frac{l(v_2) +  \alpha_1}{2}$ \\
$\delta(u,w) = H(u,w_1) + \frac{l(w_2) + l(w_3) + \alpha_2}{2}$ \\
$\delta(w,v) = H(w_1,v_1) + H(w_2,v_2) + \frac{l(w_3) + \alpha_3}{2}$ , where each $\alpha_i \in \{0,1\}$.  Then we have
\[ D=[H(u,w_1) + H(w_1,v_2) - H(u,v_1)] + H(w_2,v_2) + l(w_3) + \frac{\alpha_2 + \alpha_3 - \alpha_1}{2}. \]
Again, it suffices to show the third term $L= \frac{\alpha_2 +
\alpha_3 - \alpha_1}{2}$ is nonnegative.  An argument analogous to
the previous cases shows that $L \geq 0$.

Hence, in all cases $\delta = d_2$ satisfies the triangle inequality (and obviously symmetry and reflexivity) so it is an integer-valued metric, and it is obvious that it is Hamming compatible.
\end{proof}

\begin{rem}
It was asked by the second author in \cite{E} if the metric $T$ in Example 2.3 is minimal in its class.  That is to say, if $\delta$ is a Hamming compatible metric, must it always be the case that $\delta \geq T$?  As we have just seen, $d_2$ is such a metric that satisfies $d_2 \leq T$ by definition, with strict inequality holding for appropriate pairs of words.  However, as we will see in section 4, there exist Hamming compatible metrics that take values even smaller than $d_2$.
\end{rem}

\section{Cardinality of a sphere}
Spheres with respect to the Hamming distance are related to the
concept of error correcting codes. Define the sphere of radius $r$
centered at $u$ to be $S_r(u)= \{ v \in \Sigma^*| d_2(u,v)=r \}$. We
also define the following sphere for words of fixed length,
$S_r^j(u) = \{ v \in \Sigma_j | d_2(u,v)=r \}$, where $\Sigma_j$ is
the set of all $j$-letter words in $\Sigma^*$.

Here, we are aiming to compute the cardinality of the sphere
$S_r(u)$.

\begin{lem}
Suppose $u$ is a word of length $k$.  Then $|S_r(u)| = \displaystyle \sum_{j=k-2r}^{k+2r} |S_r^j(u)|$.
\end{lem}

\begin{proof}
If $v$ is a word of length $j$ with $d_2(u,v)=r$, then we must have $k-2r \leq j \leq k+2r$, else $\gamma_2(u,v)>r$.  Since
$S_r(u)$ is a disjoint union of these $S_r^j(u)$, the equality follows immediately.
\end{proof}

Now we compute the cardinality of $S_r^j(u)$ for a word $u$ of length $k$.  Let $v$ be a word of length $j$ with
$d_2(u,v)=r$.  Let $a=\gamma_2(u,v)$, i.e. $a$ is the smallest integer such that $|k-j| \leq 2a$, i.e. $a= \lceil \frac{|k-j|}{2} \rceil$ is the roundup of half $|k-j|$.  Then $H(\underline{u},\underline{v})=r-a$, so the concatenations must differ in exactly $r-a$ places.  If $j \leq k$, then there are ${j \choose r-a} (N-1)^{r-a}$ such $j$-letter words that differ from $u$ in exactly $r-a$ places.  If $j>k$, then there are ${k \choose r-a} (N-1)^{r-a} N^{j-k}$ such $j$-letter words. Hence, we have proved

\begin{lem}
Suppose $u$ is a word of length $k$.  Then \[ |S_r^j(u)| = \left \{
\begin{array} {c l}
{j \choose r-a} (N-1)^{r-a} & \text{if $j \leq k$} \\
\\
{k \choose r-a} (N-1)^{r-a} N^{j-k} & \text{if $j>k$}
\end{array} \right. \]
where $a= \lceil \frac{|k-j|}{2} \rceil$.
\end{lem}

\section{minimality of the $d_2$ metric}

Here we explore lower bounds of a Hamming compatible metric $\delta$
on a language $\Sigma^*$.  Recall that $\varepsilon$ is the empty
word.

\begin{df}We say that two words $u,v \in \Sigma^*$ of the same length
$n=l(u)=l(v)$ are \emph{Hamming opposites} if the Hamming distance
between them is maximal, i.e. $H(u,v)=n$.\end{df}

\begin{rem}Let $\Sigma$ be an alphabet. If $|\Sigma|=k$ and $u\in\Sigma^*$ with length $n\geq 1$, then $u$ has exactly $(k-1)^n$ Hamming opposite(s).\end{rem}

\begin{df}
We say that $\delta$ is \emph{weakly uniform}  if for any Hamming
opposites $u,v$ we have that $\delta(u,\varepsilon)=
\delta(v,\varepsilon)$.
\end{df}

\begin{prop}
If $\delta$ is weakly uniform, then $\delta(u,\varepsilon) \geq
\frac{l(u)}{2}$ for all $u \in \Sigma^*$.  In particular,
$\delta(u,\varepsilon) \rightarrow \infty$ as $l(u) \rightarrow
\infty$.
\end{prop}

\begin{proof}
Let $u$ be a word of length $l(u)=n$, and choose a Hamming opposite $v$.  Then the triangle inequality gives
\[ n = H(u,v) = \delta(u,v) \leq \delta(u,\varepsilon) + \delta(v,\varepsilon) = 2 \delta(u,\varepsilon), \]
hence $\delta(u,\varepsilon) \geq \frac{n}{2}$.
\end{proof}

\begin{rem}\label{Rem4.4}
The above argument shows that even when $\delta$ is not weakly
uniform, given any word $u$ of length $n$ and any Hamming opposite
$v$, then at least one of $u$ or $v$ must satisfy the above
inequality, since either $\delta(u,\varepsilon) \geq
\delta(v,\varepsilon)$ or $\delta(u,\varepsilon) \leq
\delta(v,\varepsilon)$. Moreover, it is important to note that this
inequality is sharp, since our weakly uniform metric $d_2$ above
satisfies $d_2(u,\varepsilon)=\lceil \frac{l(u)}{2} \rceil$.
\end{rem}

Denote by $H=H(\underline{u},\underline{v})$ the Hamming distance of the concatenation of $u,v \in \Sigma^*$.  Given a
Hamming compatible metric $\delta$, set $\gamma = \delta - H$.  Thus, we may write $\delta = H + \gamma$
where $H$ is the concatenated Hamming distance.  Observe that $\delta$ is Hamming compatible if and only if $\gamma(u,v)=0$ whenever $l(u)=l(v)$.

Since $H(\underline{u},\underline{\varepsilon})=0$ implies
$\delta(u,\varepsilon)=\gamma(u,\varepsilon)$, requiring that
$\delta$ be weakly uniform is equivalent to requiring that
$\gamma(u,\varepsilon)=\gamma(v,\varepsilon)$ for all Hamming
opposites $u,v$.

\begin{df}
Given $\delta = H + \gamma$ as above, we say that $\delta$ is \emph{uniform} if given any pair of Hamming opposites $u,v \in \Sigma_n$, we have $\gamma(u,w)=\gamma(v,w)$ for all $w \in \Sigma^*$.
\end{df}

\begin{ex}
$T(u,v):=H(\underline{u},\underline{v}) + |l(u)-l(v)|$, so here $\gamma = |l(u)-l(v)|$.
\end{ex}

\begin{ex}
$d_2(u,v)=H(\underline{u},\underline{v}) + \lceil \frac{|l(u)-l(v)|}{2} \rceil$, so here $\gamma=\lceil \frac{|l(u)-l(v)|}{2} \rceil$.
\end{ex}

Whereas weak uniformity is the simple notion that Hamming opposites
are equidistant from the empty word $\varepsilon$, this definition
of uniformity is slightly more mysterious.  If we fix
$w=\varepsilon$ this condition reduces to weak uniformity.   Notice
both examples above are uniform, and more generally, if
$\gamma(u,w)$ depends only on the lengths $l(u), l(w)$ of the inputs
$u,w$, then $\delta$ is uniform.

\begin{lem}\label{4.8}
Suppose we are given words $u,w \in \Sigma^*$ with $l(u) \geq l(w)$.  Then there exists a Hamming opposite $v$ for $u$ such that \[ H(\underline{u},\underline{w}) + H(\underline{v},\underline{w}) = l(w).\]
\end{lem}

\begin{proof}
Write $u=u_1u_2 \ldots u_n$ and $w=w_1w_2 \ldots w_m$, where $m \leq n$ by assumption.  Then let $k=H(\underline{u},\underline{w})=H(u',w)$, where $u'=u_1u_2 \ldots u_m$.  Now choose $v_i$ by the following method.  For $1 \leq i \leq m$, set $v_i = w_i$ if $w_i \neq u_i$.  If $w_i = u_i$, then take $v_i$ to be any letter in $\Sigma$ different from $u_i$.  For $m < i \leq n$, just take $v_i$ to be any letter in $\Sigma$ different from $u_i$.

Now, it is clear by the construction that $v=v_1v_2 \ldots v_n$ is a Hamming opposite for $u$ since they differ in each letter place.  Let $v'=v_1v_2 \ldots v_m$, so that $H(\underline{v},\underline{w})=H(v',w)$.  Since $u'$ and $w$ differ in $k$ places, and $v'$ agrees with $w$ in precisely those $k$ places and differs elsewhere, $H(v',w)=m-k$.  Hence, we have
\[ H(u',w) + H(v',w) = k + (m-k) = m =l(w). \]
\end{proof}

The following theorem is the main result.

\begin{thm}\label{Th4.9}
Suppose $\delta$ is a uniform Hamming compatible metric on a language $\Sigma^*$.  Then $\delta(u,v) \geq d_2(u,v)$ for all $u,v \in \Sigma^*$.  In particular, $\delta(u,v) \rightarrow \infty$ as $|l(u)-l(v)| \rightarrow \infty$.
\end{thm}

\begin{proof}
We must show that for all $u,w \in \Sigma^*$, $\delta(u,w) \geq d_2(u,w)$.  Writing $\delta=H+\gamma$, this is equivalent to showing that $\gamma(u,w) \geq \lceil \frac{|l(u)-l(w)|}{2} \rceil$.  Assume without loss of generality that $l(u) \geq l(w)$, and choose a Hamming opposite $v$ for $u$ that satisfies the equality of the above lemma.  Now, we have
\[ l(u)=H(u,v)=\delta(u,v) \leq \delta(u,w) + \delta(v,w) = H(\underline{u},\underline{w}) + \gamma(u,w) + H(\underline{v},\underline{w}) + \gamma(v,w). \]
By our choice of $v$, $H(\underline{u},\underline{w}) + H(\underline{v},\underline{w})=l(w)$, so
\[ l(u) \leq l(w) + \gamma(u,w) + \gamma(v,w) = l(w) + 2 \gamma(u,w) \]
by the uniformity assumption.  Hence, $\gamma(u,w) \geq \frac{l(u)-l(w)}{2}$.  Because $\gamma$ is integer valued, it follows that $\gamma(u,w) \geq \lceil \frac{|l(u)-l(w)|}{2} \rceil$, which completes the proof.
\end{proof}

\begin{rem}\label{Rem4.10}We remark that given any $u,w \in \Sigma^*$ and a Hamming opposite $v$ for $u$ that satisfies the equality of the
above lemma, either $\delta(u,w) \geq d_2(u,w)$ or $\delta(v,w) \geq d_2(v,w)$ must happen, since either
$\gamma(u,w) \geq \gamma(v,w)$ or $\gamma(u,w) \leq \gamma(v,w)$, \emph{without any uniformity assumption on
$\delta$}.  This means that even the wildest Hamming compatible distances must grow to some extent with $|l(u)-l(v)|$.  To make this slightly more rigorous, given any word $w \in \Sigma^*$, we can always find a sequence of words $u_n \in \Sigma_n$ such that $\delta(u_n,w) \geq d_2(u_n,w)$, and in particular $\delta(u_n,w) \rightarrow \infty$ as $n \rightarrow \infty$.  The uniformity assumption simply ensures that this growth is literally uniform.
\end{rem}

\begin{ex}
Uniformity is necessary for the above minimality result to hold.  Let $\Sigma=\{0,1\}$ and define
\[ \delta(u,\varepsilon) = \left \{
\begin{array} {c l}
{0} & \text{if $u=\varepsilon$}\\
\\
{1} & \text{if $u=000$} \\
\\
{d_2(u,000)} & \text{otherwise}
\end{array} \right. \]
where $d_2$ is the metric defined previously.  Then take
$\delta(u,v)=d_2(u,v)$ for all $u,v \in \Sigma^*-\{\varepsilon\}$.
We first show $\delta$ is a metric.  Since we have only changed
distances to $\varepsilon$, and $d_2$ has been shown to be a metric,
we only need to check the triangle inequality for expressions
involving $\varepsilon$.  First consider
\[ \delta(u,\varepsilon) \leq \delta(u,v) + \delta(v,\varepsilon). \]
If $u=\varepsilon, 000$ then this is trivial. If $u\neq \varepsilon,
000$ then this becomes
\[ d_2(u,000) \leq \delta(u,v) + \delta(v,\varepsilon). \]
If $v=\varepsilon, 000$ then again this is trivial, so assuming $v
\neq \varepsilon, 000$ this becomes
\[ d_2(u,000) \leq d_2(u,v) + d_2(v,000) \]
which is true because $d_2$ is a metric.  Now consider the expression
\[ \delta(u,v) \leq \delta(u,\varepsilon) + \delta(v,\varepsilon). \]
Again, if $u=\varepsilon, 000$ then this is trivial, and by symmetry
the same is true if $v=\varepsilon, 000$. If neither $u$ nor $v$ are
these words, then this becomes
\[ d_2(u,v) \leq d_2(u,000) + d_2(v,000) \]
which is again true because $d_2$ is a metric.  Hence, $\delta$ is a Hamming compatible integer valued metric.

It is easy to see that $\delta$ is not even weakly uniform, since
$\delta(0,\varepsilon)=d_2(0,000)=1$ and
$\delta(1,\varepsilon)=d_2(1,000)=2$.  Morever, by construction
$\delta(000,\varepsilon)=1<2=\lceil{\frac{3}{2}}
\rceil=d_2(000,\varepsilon)$, so it takes values smaller than $d_2$.

\end{ex}

\begin{ex}
The following variant of the above example shows that there exist weakly uniform metrics which are not uniform.
As before, let $\Sigma = \{0,1 \}$ and define

\[ \delta(u,0) = \left \{
\begin{array} {c l}
{0} & \text{if $u=0$}\\
\\
{1} & \text{if $u=\varepsilon, 1, 11$} \\
\\
{T(u,11)} & \text{otherwise}
\end{array} \right. \]
where $T$ is the metric defined previously.  Then take
$\delta(u,v)=T(u,v)$ for all $u,v \in \Sigma^*-\{0 \}$. We first
show $\delta$ is a metric.  Again, since we have only changed
distances to $0$, we have not changed the distance from
$\varepsilon$ to $0$ or from $1$ to $0$, and $T$ is a metric, we
only need to check the triangle inequality for expressions involving
$0$. First consider
\[ \delta(u,0) \leq \delta(u,v) + \delta(v,0). \]
If $u=\varepsilon, 0, 1, 11$ then this is trivial. If $u\neq
\varepsilon, 0, 1, 11$ and then this becomes
\[ T(u,11) \leq \delta(u,v) + \delta(v,0). \]
Now we may assume $l(u) \geq 2$.  If $v=\varepsilon$ this is
$T(u,11)\leq l(u) +1$, which is true since $T(u,11) =
H(\underline{u},11)+ l(u) -2 \leq l(u)$.  If $v=0$ this is $T(u,11)
\leq \delta(u,0)+0=T(u,11)$.  If $v=1$ this is $T(u,11) \leq
\delta(u,1)+1$.  This holds, because the left hand side is
\[H(\underline{u},11)+l(u)-2 \leq H(\underline{u},1) + l(u)-1+1\]
the right hand side.  If $v=11$ this is $T(u,11)\leq
\delta(u,11)+1=T(u,11)+1$.  Now, assuming $v \neq \varepsilon, 0, 1,
11$ this becomes
\[ T(u,11) \leq T(u,v) + T(v,11) \]
which is true because $T$ is a metric.

 Now consider the expression
\[ \delta(u,v) \leq \delta(u,0) + \delta(v,0). \]
If $u=0$ this is trivial.  If $u=1$, this becomes
\[ \delta(1,v) \leq 1+\delta(v,0). \]
If $v=\varepsilon, 0, 1$ this is trivial, and if $v=11$ this is
$1=T(1,11)\leq 1+1$.  If $v \neq \varepsilon, 0, 1, 11$ then $l(v)
\geq 2$ and this becomes $T(1,v) \leq 1+T(v,11)$.  This is
\[ H(\underline{v},1) + l(v)-1 \leq 1+H(\underline{v},11)+l(v)-2\]
which holds because $H(\underline{v},1) \leq H(\underline{v},11)$
when $v \neq 11$.  Now suppose $u=\varepsilon$.  Then our inequality
becomes $\delta(\varepsilon,v) \leq 1+ \delta(v,0)$.  If
$v=\varepsilon, 0, 1, 11$ this is obvious, so assume $v \neq
\varepsilon,0,1,11$, so $l(v) \geq 2$, and this is $ l(v) \leq
1+T(v,11)=1 + H(\underline{v},11)+l(v)-2.$ But this is just saying
that $H(\underline{v},11) \geq 1$, which is true since $v \neq 11$
and $l(v) \geq 2$.

If $u=11$ then our inequality becomes $\delta(11,v) \leq 1+
\delta(v,0)$.  If $v=\varepsilon,0,1,11$ this is clear, and if $v
\neq \varepsilon,0,1,11$ this is just saying that $T(11,v) \leq
1+T(v,11)$ which is trivial.  The final case is when $u,v \neq
\varepsilon,0,1,11$, in which case our inequality becomes
\[ T(u,v) \leq T(u,11) + T(v,11) \]
which is again true because $T$ is a metric.  Hence, $\delta$ is a Hamming compatible integer valued metric.

Now, $\delta$ is weakly uniform because
$\delta(0,\varepsilon)=\delta(1,\varepsilon)=1$, and for $l(u)>1$,
we have
$\delta(u,\varepsilon)=T(u,\varepsilon)=T(v,\varepsilon)=\delta(v,\varepsilon)$
for all $u,v \in \Sigma_n$.  However, $\delta$ is not uniform
because $\gamma(11,0)=\delta(11,0) -
H(\underline{11},\underline{0})=1-1=0$, while
$\gamma(00,0)=\delta(00,0) - H(\underline{00},\underline{0})=2-0=2$.
\end{ex}

\section{Non-uniform metrics}

These examples illustrate that arbitrary Hamming compatible metrics may behave wildly in general, so some uniformity condition is necessary in
order to say something about minimality.  Furthermore, unless there are distinguished words, one desires such uniformity for this type of metric.

However, given an arbitrary Hamming compatible metric $\delta$, one may na\"ively
 ask the question, how many words of a given length must satisfy the inequality of Theorem~\ref{Th4.9}?
 Suppose we have words $u \in \Sigma_n$, $w \in \Sigma_m$ with $n \geq m$ that violate the inequality of the theorem, i.e. $\delta(u,w) < d_2(u,w)$.
 In view of Remark~\ref{Rem4.10}, any Hamming opposite $v$ of $u$ chosen by the method of Lemma~\ref{4.8} must then satisfy $\delta(v,w) \geq d_2(v,w)$.

Let $h:=H(\underline{u},\underline{w})$ and $N=|\Sigma|$ as before.
There are $(N-1)^{n}$ total Hamming opposites of $u$, and of these
there are $(N-1)^{m-h}(N-1)^{n-m}=(N-1)^{n-h}$ Hamming opposites
that satisfy the equality of Lemma~\ref{4.8}, and hence the
inequality of Theorem~\ref{Th4.9}.  For instance, if $w=
\varepsilon$, then $h=0$ and so every Hamming opposite $v$ of $u$
must satisfy $\delta(v,\varepsilon) \geq d_2(v,\varepsilon)=\lceil
\frac{n}{2} \rceil$, which recovers the first part of
Remark~\ref{Rem4.4}.

In general, the probability of a Hamming opposite satisfying the inequality of the theorem is at least $\frac{(N-1)^{n-h}}{(N-1)^n}= \frac{1}{(N-1)^h} \geq \frac{1}{(N-1)^m}$.  The exceptional case is when $\Sigma=\{0,1\}$, i.e. $N=2$, because in this setting Hamming opposites are unique.  Here, if $\delta(u,w) < d_2(u,w)$, then we must have $\delta(v,w) \geq d_2(v,w)$ for the unique Hamming opposite $v$ of $u$,
so at least half of the words of a given length must satisfy the inequality of the theorem.

\end{document}